\documentclass[pra,twocolumn]{revtex4-2}
\usepackage{amsmath}
\usepackage{amssymb}
\usepackage{graphicx}
\usepackage{braket}
\usepackage{stmaryrd}
\usepackage{hyperref}
\usepackage{color}
\usepackage{dsfont}
\usepackage{bm}
\usepackage{booktabs}
\usepackage{amsthm}

\hypersetup{
    colorlinks,
    citecolor=blue,
    linkcolor=blue,
    urlcolor=blue
}

\newtheorem{lemma}{Lemma}

\begin{document}

\title{Structural $f$-divergence: Tight Universal Bounds for Cost Function Moments and Gradients in Parameterized Quantum Circuits}

\author{Tomohiro Nishiyama}
\email{htam0ybboh@gmail.com}
\affiliation{Independent Researcher, Tokyo 206-0003, Japan}

\author{Yoshihiko Hasegawa}
\email{hasegawa@biom.t.u-tokyo.ac.jp}
\affiliation{Department of Electrical Engineering and Information Systems, Graduate School of Engineering, The University of Tokyo,
Tokyo 113-8656, Japan}
\date{\today}

\begin{abstract}

The barren plateau phenomenon, in which cost-function gradients of 
variational quantum algorithms vanish exponentially, remains a central obstacle 
for near-term quantum computing. Existing analyses typically depend on 
t-design or Haar-random assumptions and bound quantities at the level of 
unitary distributions, offering limited insight for designing probability 
measures on the parameter space of parameterized quantum circuits. 
In this paper, we introduce the structural $f$-divergence, a symmetric $f$-divergence-based measure between probability distributions on the parameter space. 
We establish analytically  trade-off inequalities that bound the discrepancies in the 
expected gradient magnitude and in the cost-function moments between a
distribution on PQC
and a reference distribution; equality is attained by a 
minimal one-qubit, one-layer ansatz. As applications, we derive necessary 
conditions on probability measures for avoiding BPs and cost concentration, 
and sufficient conditions that suppress noise-induced deviations.

\end{abstract}

\maketitle
\section{Introduction\label{sec:Introduction}}

In recent years, variational quantum algorithms (VQAs) have emerged as leading candidates for demonstrating the utility of quantum computing on noisy intermediate-scale quantum (NISQ) devices. However, the optimization process of VQAs faces a significant challenge known as the barren plateau (BP) phenomenon, where the gradients of the cost function with respect to the circuit parameters vanish exponentially. In Ref.~\cite{mcclean2018barren}, the origins of BPs have been extensively analyzed from various perspectives. More recently, algebraic approaches utilizing Lie algebras have enabled a more rigorous discussion of the relationship between circuit symmetry and gradient vanishing~\cite{ragone2024lie, diaz2023showcasing}.
To mitigate the BP problem, numerous strategies have been proposed, including the adoption of local cost functions~\cite{cerezo2021cost}, sophisticated initialization schemes~\cite{grant2019initialization}, and layer-by-layer optimization of circuit parameters ~\cite{skolik2021layerwise}. A common thread underlying these strategies is the management of the quantum circuit's expressibility. Reference~\cite{sim2019expressibility} proposed a quantitative measure of expressibility based on the Kullback-Leibler divergence, using the Haar random distribution as an ideal reference. Furthermore, Ref.~\cite{holmes2022connecting} introduced a metric based on the diamond norm, elucidating a fundamental trade-off: circuits with higher expressibility (those closer to Haar randomness) are more susceptible to the risk of encountering BPs. 
Furthermore, the BP phenomenon is closely linked to the cost concentration around a specific mean. It has been shown that these two effects are equivalent under certain conditions, such as when the circuit depth is sufficiently large~\cite{arrasmith2022equivalence}. However, most existing studies rely on analyses that assume a high degree of randomness, such as t-designs or Haar distributions. In practical applications, circuits often operate in transient regimes that do not reach a t-design or involve specific probability distributions. Despite this, there is still a lack of comprehensive and unified research on how the discrepancy from a reference distribution quantitatively impacts the behavior of gradients and higher-order moments in such realistic scenarios. Moreover, conventional evaluations primarily focus on bounds derived from the distribution of unitary operators. Consequently, they fail to provide direct design principles for the probability measures over the parameter space in actual circuit architecture. Furthermore, existing inequalities regarding the trade-off between expressibility and BPs remain insufficient, as they do not offer tight evaluations that specify the conditions under which the equality holds.

In this paper, we introduce the \textit{structural $f$-divergence} [cf. Eq.~\eqref{eq:def_ex_div}], a measure defined between probability measures on the parameter space or, more essentially, their induced measures on the unitary group manifold.
The structural $f$-divergence is formulated based on the symmetric $f$-divergence, encompassing a canonical class of metrics including the total variation distance, Jensen-Shannon divergence, and squared Hellinger distance. Our framework requires only three fundamental conditions on the generator function: twice-differentiability, strict convexity, and the normalization $f(1)=0$.
A key contribution of this work is the derivation of fundamental trade-off relations between this structural $f$-divergence and the statistical characteristics of quantum circuits. Specifically, we establish rigorous bounds on the discrepancies in the expected magnitude of gradients [cf. Eq.~\eqref{eq:main_result1}] and the moments of the cost function [cf. Eq.~\eqref{eq:main_result2}] when comparing a distribution on PQCs
against a reference measure. Notably, these bounds are analytically tight; we demonstrate that the equalities are exactly attained by a minimal configuration: a one-qubit, one-layer ansatz under two-element probability measures. 
Establishing bounds on the parameter space offers a significant advantage: it enables a direct evaluation of how specific sampling strategies, such as initialization schemes or parameter-wise noise, influence the resulting cost function landscape.
These inequalities provide the sharpest possible characterization of how a circuit's statistical profile constrains its trainability. 
As practical applications of our results, we first derive necessary conditions on probability measures to avoid barren plateaus and cost concentration. Furthermore, we establish sufficient conditions for suppressing deviations in the expectation of the absolute value of gradients and the moments of the cost function when noise-induced perturbations shift the probability measure away from the ideal case.
Our results reveal a universal connection between the information-geometric structure of the underlying distributions and the observable landscape behavior, independent of the specific functional form of the objective.

\begin{table*}
    \centering
    \begin{tabular}{ccc}
        \toprule
        & Bound & Equality conditions  \\
        \midrule
        Magnitude of gradient expectation value  & $\displaystyle  \frac{\left|\mathbb{E}_{P_{\Theta}}[\:|\partial_j \braket{O}|\:]-\mathbb{E}_{Q_{\Theta}}[\:|\partial_j \braket{O}|\:]\right|}{\|O\|_\infty} \le 2\|H_j\|_RD_{f}^{\mathrm{str}}(P_{\Theta},Q_{\Theta})$    & Eqs.~\eqref{eq:equality_cond_circuit} and~\eqref{eq:equality_cond_grad}\\
        Moment of cost function ($k$: even) & $\displaystyle \frac{\left|\mathbb{E}_{P_{\mathcal{Z}}}[\braket{O}^k]-\mathbb{E}_{Q_{\mathcal{Z}}}[\braket{O}^k]\right|}{\|O\|_\infty^k }\le  D_{f}^{\mathrm{str}}(P_{\mathcal{Z}},Q_{\mathcal{Z}})$ &  Eqs.~\eqref{eq:equality_cond_circuit} and~\eqref{eq:equality_cond_k_even} (\text{when $\mathcal{Z}=\mathcal{U}$})\\
        Moment of cost function ($k$: odd) & $\displaystyle \frac{\left|\mathbb{E}_{P_{\mathcal{Z}}}[\braket{O}^k]-\mathbb{E}_{Q_{\mathcal{Z}}}[\braket{O}^k]\right|}{\|O\|_\infty^k }\le  2D_{f}^{\mathrm{str}}(P_{\mathcal{Z}},Q_{\mathcal{Z}})$ &  Eqs.~\eqref{eq:equality_cond_circuit} and~\eqref{eq:equality_cond_k_odd} (\text{when $\mathcal{Z}=\mathcal{U}$}) \\
        \bottomrule
    \end{tabular}
    \caption{Summary of results. $P_{\Theta}$ and $Q_{\Theta}$ are probability measures on the parameter space $\Theta$, and $P_{\mathcal{U}}$ and $Q_{\mathcal{U}}$ are their push-forward measures on the unitary group. We let the index $\mathcal{Z}$ range over  $\mathcal{Z}:=\{\Theta,\mathcal{U}\}$ to refer to either $(P_{\Theta},Q_{\Theta})$ and $(P_{\mathcal{U}},Q_{\mathcal{U}})$. For probability measures $P$ and $Q$, $\mathbb{E}_P[\bullet]$ and $\mathbb{E}_Q[\bullet]$ denote expectation values with respect to $P$ and $Q$, respectively. 
$ D_{f}^{\mathrm{str}}(P,Q)$ is the structural $f$-divergence that depends on the function $f$.
$O$ is a objective function and $\|O\|_\infty$ denotes the operator norm. $H_j$ is a generator of the unitary operator $U(\theta_j)$, and $\|H\|_R$ is defined as $(\lambda_H^{\max}-\lambda_H^{\min})/2$, where $\lambda_H^{\max}$ and $\lambda_H^{\min}$ denote the maximum and minimum eigenvalues of $H$.}
    \label{tab:results}
\end{table*}

\begin{figure}
\centering
\includegraphics[width=1\linewidth]{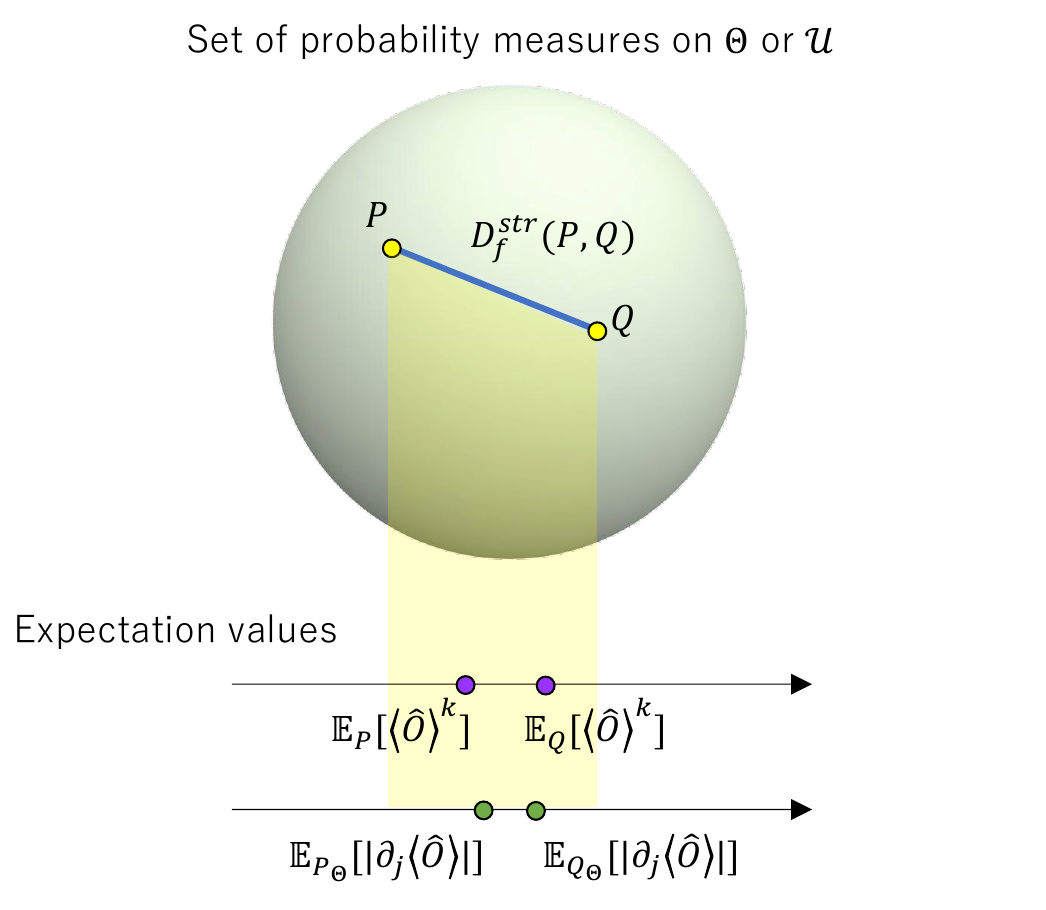}
\caption{
Conceptual diagram illustrating that the structural $f$-divergence $D_f^{\mathrm{str}}(P,Q)$ bounds both the difference in the $k$-th moments of the cost function and the difference in the expected gradient magnitudes between probability measures $P$ and $Q$ on either the parameter space $\Theta$ or the unitary group $\mathcal{U}$. Here, $\braket{O}$ denotes the cost function defined in Eq.~\eqref{eq:def_cost_func} and $\hat{O}=O/\|O\|_\infty$ denotes the operator normalized by the operator norm.
}
\label{fig:bound}
\end{figure}

\section{Preliminaries}

\subsection{Parametrized quantum circuits}

PQCs are quantum circuits with tunable gate parameters, usually used in algorithms where the parameters are optimized to solve a task.
The full circuit structure $U(\bm{\theta})$ is composed of alternating layers of gates $\{U_l(\bm{\theta}_l)\}$.
For a circuit comprising $n$ qubits and $L$ layers, the unitary operator $U(\bm{\theta})$ is defined as
\begin{align}
    U(\bm{\theta})=\prod_{l=1}^L U_l(\bm{\theta}_l)=\prod_{l=1}^L \prod_{m=1}^M e^{-i \theta_{l,m} H_{l,m}}.
    \label{def:unitary_gates}
\end{align}
Here, $\theta_{l,m}$ denotes the parameter (which can be either a variational variable or a fixed value), and $H_{l,m}$ is the corresponding Hermitian generator for the $m$-th gate in layer $l$.

The circuit starts in the quantum state $\ket{\mathrm{init}}$ and ends with measurements that estimate expectation values of observables used for optimization.
The cost function, computed from these measurement outcomes and used as the optimization objective, is
\begin{align}
    \braket{O}_{\bm{\theta}}:=\mathrm{Tr}\!\left[\rho\, U(\bm{\theta})^\dagger O\, U(\bm{\theta})\right],
    \label{eq:def_cost_func}
\end{align}
where $O$ is a Hermitian operator and $\rho:=\ket{\mathrm{init}}\bra{\mathrm{init}}$.
For notational simplicity, the subscript $\bm{\theta}$ will be omitted where no confusion can arise.

\subsection{$f$-divergence}

Let $p:=dP/d\mu$ and $q:=dQ/d\mu$ be the densities of probability measures $P$ and $Q$ with respect to a common dominating measure $\mu$.
Let $f:(0,\infty)\rightarrow \mathbb{R}$ be a convex function satisfying $f(1)=0$. The $f$-divergence~\cite{sason2016f} between probability measures $P$ and $Q$ is defined as 
\begin{align}
    D_f(P\|Q):=\int qf(p/q) d\mu,
    \label{eq:def_f_div}
\end{align}
where we adopt the convention $0f(x/0):=0$.
We further define the symmetric $f$-divergence as 
\begin{align}
    \tilde{D}_f(P,Q):=\frac{1}{2}(D_f(P\|Q)+D_f(Q\|P)).
    \label{eq:def_sym_f_div}
\end{align}
This symmetric divergence corresponds to the $f$-divergence associated with the function
\begin{align}
    \tilde{f}(x):=\frac{1}{2}\left[f(x)+xf\left(\frac{1}{x}\right)\right].
    \label{eq:sym_f}
\end{align}
Let $r\in [-1,1]$ and $s:=|r|$ be variables. 
We define two-element probability measures:
\begin{align}
    P_B:=\left\{\frac{1-r}{2}, \frac{1+r}{2}\right\}, \nonumber\\
    Q_B:=\left\{\frac{1+r}{2}, \frac{1-r}{2}\right\}. 
    \label{eq:def_two_element}
\end{align}
Let $d_f:[0,1]\rightarrow[0,\infty)$ be the symmetric $f$-divergence between two-element probability measures:
\begin{align}
    &d_f(s):=\tilde{D}_f(P_B,Q_B)\nonumber\\
    &=\frac{(1-s)}{2}f\left(\frac{1+s}{1-s}\right)+\frac{(1+s)}{2}f\left(\frac{1-s}{1+s}\right).
    \label{eq:def_binary_div}
\end{align}
Since the function $d_f$ is a monotonically increasing function (see Lemma~\ref{lem1} in Appendix~\ref{sec:proof_properties}), there exists an inverse function $d_f^{-1}$.

\section{Results}

\subsection{Structural $f$-divergence}

Let $\mathcal{F}$ be a set of functions such that
\begin{align}
    \mathcal{F}:=\{f\in C^2(0,\infty)\;|f(1)=0, f''>0\} \cup \left\{\frac{|x-1|}{2}\right\}.
    \label{eq:def_set_function}
\end{align}
Here, the function $f(x)=|x-1|/2$ corresponds to the total variation distance $d_{\mathrm{TV}}(P,Q):=1/2\int|p-q|d\mu$.
Let $P_{\Theta}$ and $Q_{\Theta}$ be probability measures on the parameter space $\Theta$.
Through the mapping defined in Eq.~\eqref{def:unitary_gates}, these measures induce the respective push-forward measures $P_{\mathcal{U}}$ and $Q_{\mathcal{U}}$ on the subset of the unitary group $\mathcal{U}(d)$ of degree $d=2^n$. For any function $h$, the push-forward measure satisfies
\begin{align}
    \int_\mathcal{U} h(U)dP_{\mathcal{U}}(U) =\int_\Theta h(U(\bm{\theta}))dP_{\Theta}(\bm{\theta}).
    \label{eq:induced_measure}
\end{align}
Throughout this paper, the subscript $\Theta$ or $\mathcal{U}$ is dropped whenever the statement holds for both $(P_{\Theta},Q_{\Theta})$ and $(P_{\mathcal{U}},Q_{\mathcal{U}})$.

To quantify the discrepancy between the probability measures $(P_{\Theta},Q_{\Theta})$ or $(P_{\mathcal{U}},Q_{\mathcal{U}})$,
we introduce the \textit{structural $f$-divergence} for $f\in\mathcal{F}$:
\begin{align}
    D_{f}^{\mathrm{str}}(P,Q):=d_f^{-1}(\tilde{D}_f(P,Q)).
    \label{eq:def_ex_div}
\end{align}
We adopt this definition involving $d_{f}^{-1}$ for several reasons (which will be discussed in detail later), primarily to ensure the divergence is normalized within $[0, 1]$ and to simplify the bounds in our subsequent main results.
When $f(x)=|x-1|/2$, we define $ D_{f}^{\mathrm{str}}(P,Q):=d_{\mathrm{TV}}(P,Q)$ from $d_f(s)=s$.
The structural $f$-divergence satisfies the following properties (see Appendix~\ref{sec:proof_properties}).
\begin{align}
    &0\le D_{f}^{\mathrm{str}}(P,Q)\le 1, \nonumber \\
    & D_{f}^{\mathrm{str}}(P,Q)=0,\; \text{if and only if} \; P=Q,\nonumber \\
    & D_{f}^{\mathrm{str}}(P,Q)= D_{f}^{\mathrm{str}}(Q,P), \nonumber\\
    & D_{f}^{\mathrm{str}}(P_{\mathcal{Z}},Q_{\mathcal{Z}})\geq D_{f}^{\mathrm{str}}(P_{\mathcal{Y}},Q_{\mathcal{Y}}).
    \label{eq:properties_div}
\end{align}
In the last inequality, the transition from $\mathcal{Z}=\{\Theta,\mathcal{U}\}$ to $\mathcal{Y}$ is characterized by a stochastic map $K$ satisfying $P_{\mathcal{Y}}(y)=\int K(y|z)dP_{\mathcal{Z}}(z)$ and $Q_{\mathcal{Y}}(y)=\int K(y|z)dQ_{\mathcal{Z}}(z)$.
The positivity and the second property are necessary properties for the divergence. The third equality shows that the structural $f$-divergence is symmetric. The fourth inequality is the data processing inequality.
Since $P_{\mathcal{U}}$ and $Q_{\mathcal{U}}$ are push-forward measures [Eq.~\eqref{eq:induced_measure}], the data processing inequality yields
\begin{align}
    D_{f}^{\mathrm{str}}(P_{\mathcal{U}},Q_{\mathcal{U}}) \le D_{f}^{\mathrm{str}}(P_{\Theta},Q_{\Theta}).
    \label{eq:data_processing_str}
\end{align}
Reference~\cite{holmes2022connecting} quantified circuit expressibility by employing the diamond norm to measure the distance between the Haar measure and the ensemble of unitary operators induced by a PQC. Based on this geometric characterization, they established fundamental trade-off relations demonstrating that increased expressibility inevitably leads to a suppression of the gradient variance. In the limit where the reference measure $Q_{\mathcal{U}}$ serves as a surrogate for the Haar measure—such as a $t$-design or a sufficiently randomized ensemble—the structural $f$-divergence acts as an expressibility-like metric that characterizes the statistical proximity of the PQC's ensemble to a random distribution.

\subsection{Tight bounds}

We show that the difference between the probability distributions of the gradient magnitude and the cost function moments is bounded by the structural $f$-divergence. Notably, these bounds are all tight.

For a function $h$, the expectation value with respect to $P_{\Theta}$ is defined as
\begin{align}
    \mathbb{E}_{P_{\Theta}}[h(\bm{\theta})] := \int h(\bm{\theta}) dP_{\Theta}(\bm{\theta}).
    \label{eq:def_expectation}
\end{align}
From Eq.~\eqref{eq:induced_measure}, the expectation value with respect to $P_{\mathcal{U}}$ satisfies 
\begin{align}
    &\mathbb{E}_{P_{\mathcal{U}}}[h(U)]:= \int h(U) dP_{\mathcal{U}}(U)=\nonumber\\
    &\int h(U(\bm{\theta})) dP_{\Theta}(\bm{\theta})=\mathbb{E}_{P_{\Theta}}[h(U(\bm{\theta}))].
    \label{eq:def_expectation}
\end{align}
For $k\in \mathbb{N}$, the $k$-th moment of the cost function is defined as
\begin{align}
    &\mathbb{E}_{P_{\Theta}}[\braket{O}^k]:=\int \mathrm{Tr}[\rho U(\bm{\theta})^{\dagger}O U(\bm{\theta})]^kdP_{\Theta}(\bm{\theta})\nonumber\\
    &=\int \mathrm{Tr}[\rho U^{\dagger}O U]^kdP_{\mathcal{U}}(U)=:\mathbb{E}_{P_{\mathcal{U}}}[\braket{O}^k].
    \label{eq:induced_relation}
\end{align}
Analogously, we consider the expectation value of the gradient with respect to $P_{\Theta}$. 
For simplicity, we define $\partial_j \braket{O}:=\partial_{\theta_{j}}\braket{O}_{\bm{\theta}}$ as the gradient corresponding to the $j$-th variational parameter $\theta_{j}$, where the index $j$ refers to the $m$-th gate in layer $l$. 
Note that the expectation value with respect to $P_{\mathcal{U}}$ cannot be defined because 
$\partial_{\theta_j}\braket{O}_{\bm{\theta}}$ is not a functional of the unitary $U(\bm{\theta})$ (see Lemma~\ref{lem4} in Appendix~\ref{sec:proof_main}).
Let $\lambda_A^{\min}, \lambda_A^{\max}$ be the minimum and maximum eigenvalues of a Hermitian operator $A$, respectively. We define $\|A\|_R:=(\lambda_A^{\max}-\lambda_A^{\min})/2$, and let $\|A\|_\infty$ be the operator norm. 
Let $I$ be the $2\times 2$ identity operator. We consider a one-qubit and one-layer ansatz $\mathcal{C}_{1,1}$ that corresponds to the case of $L=M=1$ in Eq.~\eqref{def:unitary_gates} and consists of the following components:
\begin{align}
    &H = \frac{\sigma_z}{2}, \; O=\sigma_x,\nonumber\\ 
    & \ket{\mathrm{init}}:=\frac{1}{\sqrt{2}}(\ket{0}+\ket{1}),
    \label{eq:equality_cond_circuit}
\end{align}
where $\{\sigma_{\alpha}\}$ are the Pauli matrices for $\alpha=\{x,y,z\}$ and the parameter $\theta$ is in the range $\in[0,2\pi]$.

For the structural $f$-divergence, the following two inequalities hold:
\begin{align}
    \frac{\left|\mathbb{E}_{P_{\Theta}}[\:|\partial_j \braket{O}|\:]-\mathbb{E}_{Q_{\Theta}}[\:|\partial_j \braket{O}|\:]\right|}{\|O\|_\infty}&\le 2\|H_j\|_RD_{f}^{\mathrm{str}}(P_{\Theta},Q_{\Theta}),
    \label{eq:main_result1} \\
    \frac{\left|\mathbb{E}_{P}[\braket{O}^k]-\mathbb{E}_{Q}[\braket{O}^k]\right|}{\|O\|_\infty^k }&\le C(k) D_{f}^{\mathrm{str}}(P,Q), \label{eq:main_result2} 
\end{align}
where
\begin{align}
    C(k)&:=
    \begin{cases}
        2, & (k\text{ : odd}),\\
        1, & (k\text{ : even}).
    \end{cases} 
    \label{eq:def_coeff}
\end{align}
Equation~\eqref{eq:main_result1} holds for the unitary defined in Eq.~\eqref{def:unitary_gates}, while Eq.~\eqref{eq:main_result2} holds for an arbitrary unitary $U(\bm{\theta})$.
Note that Eq.~\eqref{eq:main_result2} is valid for both $(P_{\Theta},Q_{\Theta})$ and $(P_{\mathcal{U}},Q_{\mathcal{U}})$.
Equations~\eqref{eq:main_result1} and~\eqref{eq:main_result2} are the main results of this paper. 
The proofs and the equality conditions for Eqs.~\eqref{eq:main_result1} and~\eqref{eq:main_result2} are provided in Appendix~\ref{sec:proof_main} and Appendix~\ref{sec:proof_main2}, respectively.
Within these bounds, since the right-hand side holds for any $f\in \mathcal{F}$, it can be replaced by $D^{\mathrm{str}}(P,Q):=\inf_{f\in\mathcal{F}}D_{f}^{\mathrm{str}}(P,Q)$.
For a normalized objective function $\hat{O}:=O/\|O\|_\infty$, these bounds establish that the discrepancy in expectation values for the absolute value of the gradient or the moments of the cost function is characterized by 
$\mathcal{O}\left(D^{\mathrm{str}}(P_{\Theta},Q_{\Theta})\right)$ or $\mathcal{O}\left(D^{\mathrm{str}}(P_{\mathcal{U}},Q_{\mathcal{U}})\right)$.
These results collectively imply that the structural $f$-divergence serves as a foundational metric for quantum variational landscapes. Specifically, the fact that both high-order moments and the absolute value of the gradient are tightly bounded by this divergence—independent of the specific choice of $f$
—reveals a universal connection between the information-geometric structure of the distributions and the observable behavior of the model. Since there exist ansätze for which the equality holds, our framework provides the sharpest possible characterization of how the statistical profile of a quantum circuit is constrained by the underlying probability measures.

For Eq.~\eqref{eq:main_result1}, the equality holds for the ansatz $\mathcal{C}_{1,1}$ with the two-element probability measures, regardless of the choice of $f\in\mathcal{F}$:
\begin{align}
    P_{\Theta}^B(\theta=0)=\frac{1-r}{2}, \;P_{\Theta}^B\left(\theta=\frac{\pi}{2}\right)=\frac{1+r}{2}, \nonumber \\
    Q_{\Theta}^B(\theta=0)=\frac{1+r}{2}, \;Q_{\Theta}^B\left(\theta=\frac{\pi}{2}\right)=\frac{1-r}{2}.
    \label{eq:equality_cond_grad}
\end{align}
Regarding Eq.~\eqref{eq:main_result2}, let us consider the equality condition for $(P_{\mathcal{U}},Q_{\mathcal{U}})$. 
When $k$ is even, the equality holds for the ansatz $\mathcal{C}_{1,1}$ with the two-element probability measures, which correspond to Eq.~\eqref{eq:equality_cond_grad}:
\begin{align}
    P_{\mathcal{U}}^B\left(U=I\right)=\frac{1-r}{2}, \;P_{\mathcal{U}}^B\left(U=e^{-i\pi H/2}\right)=\frac{1+r}{2}, \nonumber\\
    Q_{\mathcal{U}}^B\left(U=I\right)=\frac{1+r}{2}, \;Q_{\mathcal{U}}^B\left(U=e^{-i\pi H/2}\right)=\frac{1-r}{2}.
    \label{eq:equality_cond_k_even}
\end{align}
When $k$ is odd, the equality condition is the same as that in the even case, except that the probability measures are replaced by the following:
\begin{align}
    P_{\mathcal{U}}^B\left(U=I\right)=\frac{1-r}{2}, \;P_{\mathcal{U}}^B\left(U=e^{-i\pi H}\right)=\frac{1+r}{2}, \nonumber\\
    Q_{\mathcal{U}}^B\left(U=I\right)=\frac{1+r}{2}, \;Q_{\mathcal{U}}^B\left(U=e^{-i\pi H}\right)=\frac{1-r}{2}.
    \label{eq:equality_cond_k_odd}
\end{align}
The equality conditions for $(P_{\Theta},Q_{\Theta})$ are identical to those of Eqs.~\eqref{eq:equality_cond_k_even} and~\eqref{eq:equality_cond_k_odd}, provided that the push-forward measure $(P_{\mathcal{U}},Q_{\mathcal{U}})$ is replaced with $(P_{\Theta},Q_{\Theta})$ (see Appendix~\ref{sec:equaility_nqubit} for $n=1$). 
For simplicity, we provided examples of equality conditions for Eqs.~\eqref{eq:main_result1} and~\eqref{eq:main_result2} in the one-qubit case. However, it is also possible to construct such examples for the $n$-qubit case (see Appendix~\ref{sec:equaility_nqubit}).

We consider two examples as applications of Eqs.~\eqref{eq:main_result1} and~\eqref{eq:main_result2}. In the first example, we discuss the necessary conditions for avoiding barren plateaus or cost concentration.
Let $P_{\Theta}^{\mathrm{BP}}$ be a probability measure that exhibits a barren plateau~\cite{mcclean2018barren}, characterized by an exponentially vanishing variance of the gradient:
\begin{align}
    \mathrm{Var}_{P_{\Theta}^{\mathrm{BP}}}[\partial_j \braket{O}] \le \mathcal{O}(b^{-n}),
\end{align}
for some $b>1$, while its expectation value satisfies $\mathbb{E}_{P_{\Theta}^{\mathrm{BP}}}[\partial_j \braket{O}]= 0$. Here $\mathrm{Var}_P[\bullet]$ is the variance with respect to $P$.
Under such a measure, the typical magnitude of the gradient $\mathbb{E}_{P_{\Theta}^{\mathrm{BP}}}[\:|\partial_j \braket{O}|]$ also vanishes exponentially with $n$.
Therefore, to maintain the gradient above a threshold $g_{\mathrm{th}}>0$, the probability measure $P_{\Theta}$ must satisfy
\begin{align}
    D^{\mathrm{str}}(P_{\Theta},P_{\Theta}^{\mathrm{BP}})&\geq \frac{(g_{\mathrm{th}}-\mathbb{E}_{P_{\Theta}^{\mathrm{BP}}}[\:|\partial_j \braket{O}|\:])}{\|H_j\|_R\|O\|_\infty }\nonumber\\
    &=\frac{g_{\mathrm{th}}}{\|H_j\|_R\|O\|_\infty }+\mathcal{O}(b^{-n}).
    \label{eq:grad_thresh}
\end{align}

Similarly, we denote $P_{\Theta}^{\mathrm{CC}}$ as a probability measure that exhibits cost concentration, where the cost function $\braket{O}_{\bm{\theta}}$ concentrates exponentially around its expectation value:
\begin{align}
    \mathrm{Var}_{P_{\Theta}^{\mathrm{CC}}}[\braket{O}] \le \mathcal{O}(c^{-n}),
\end{align}
for some $c>1$. Under such a measure, the landscape of the cost function becomes increasingly flat as $n$
increases, as the probability that the cost deviates from its mean $\mathbb{E}_{P_{\Theta}^{\mathrm{CC}}}[O]$ is exponentially suppressed.
Under appropriate conditions, 
the equivalence between barren plateaus and the exponential concentration of the cost function was established in~\cite{arrasmith2022equivalence}.
Let $P^{\mathrm{CC}}_{\mathcal{U}}$ be a push-forward measure of $P_{\Theta}^{\mathrm{CC}}$.
To shift the measure $P$ away from a cost-concentrated measure $P^{\mathrm{CC}}$ such that the difference in the $k$-th moment satisfies $|\mathbb{E}_{P}[\braket{O}^k]-\mathbb{E}_{P^{\mathrm{CC}}}[\braket{O}^k]|\geq \delta$, the measure $P$ must satisfy 
\begin{align}
    & D^{\mathrm{str}}(P,P^{\mathrm{CC}})\geq \frac{\delta}{\|O\|_\infty ^k}.\label{eq:moment_thresh}
\end{align}
Here we employ Eq.~\eqref{eq:main_result2}, omitting the subscripts $\Theta$ and $\mathcal{U}$ for $P$ and $P^{\mathrm{CC}}$ as previously mentioned.

As a second example, let $P^\prime_{\Theta}$ be the noise-perturbed probability measure. We evaluate the impact of the shift from $P_{\Theta}$ to $P^\prime_{\Theta}$ on the expectation of the absolute gradient and the moments of the cost function.
Based on Eqs.~\eqref{eq:main_result1} and~\eqref{eq:main_result2}, sufficient conditions to bound the deviation of the expectation value of the absolute gradient within $g_{\mathrm{th}}$, and that of the $k$-th moment within $\delta$, are respectively given by the following inequalities for some $f\in\mathcal{F}$:
\begin{align}
    D_f^{\mathrm{str}}(P_{\Theta},P^\prime_{\Theta})&\le\frac{g_{\mathrm{th}}}{\|H_j\|_R\|O\|_\infty }, \\
    D_f^{\mathrm{str}}(P_{\Theta},P^\prime_{\Theta})&\le \frac{\delta}{\|O\|_\infty ^k}.
\end{align}
In general, as the dimension of the probability measures increases, the divergence tends to grow, making it difficult to satisfy this sufficient condition.

\subsection{Asymptotic behavior}

In this section, we consider the asymptotic behavior of Eqs.~\eqref{eq:main_result1} and~\eqref{eq:main_result2} in the limit as $P$ approaches $Q$.
Let $\Delta(P,Q)$ be the triangular discrimination~\cite{le2012asymptotic, sason2016f} defined as 
\begin{align}
    \Delta(P,Q):=\frac{1}{2}\int \frac{(p-q)^2}{p+q}d\mu.
    \label{eq:def_triangular}
\end{align}
The corresponding function $f$ is $f(x)=(x-1)^2/(2(x+1))$.
In the limit where $P$ approaches $Q$, $ D_{f}^{\mathrm{str}}(P,Q)$ reduces to $\sqrt{\Delta(P,Q)}$ for any differentiable function $f\in \mathcal{F}$ (see Appendix~\ref{sec:asymptotic}). Since $d_f(s)=s^2$ for the triangular discrimination, the structural $f$-divergence is exactly given by $ D_{f}^{\mathrm{str}}(P,Q)=\sqrt{\Delta(P,Q)}$. Therefore, Eqs.~\eqref{eq:main_result1} and~\eqref{eq:main_result2} reduce to the following bounds:
\begin{align}
     \frac{\left|\mathbb{E}_{P_{\Theta}}[\:|\partial_j \braket{O}|\:]-\mathbb{E}_{Q_{\Theta}}[\:|\partial_j \braket{O}|\:]\right|}{\|O\|_\infty}&\le 2\|H_j\|_R \sqrt{\Delta(P_{\Theta},Q_{\Theta})}, \label{eq:main1_assymp}\\
    \frac{\left|\mathbb{E}_P[\braket{O}^k]-\mathbb{E}_Q[\braket{O}^k]\right|}{\|O\|_\infty^k}&\le C(k) \sqrt{\Delta(P,Q)}.
    \label{eq:main2_assymp}
\end{align}
The second inequality holds for both $(P,Q)=(P_{\Theta},Q_{\Theta})$ and $(P_{\mathcal{U}},Q_{\mathcal{U}})$.
It should be noted that these bounds hold without the assumption that $P$ and $Q$ are close.

In cases where $P$ and $Q$ are parameterized such that $P = P(\alpha)$ and $Q = P(\alpha + \delta \alpha)$, the triangular discrimination can be written using the Fisher information $\mathcal{I}(\alpha)$ as 
\begin{align}
    \Delta(P,Q)&=\frac{(\delta \alpha)^2}{4}\mathcal{I}(\alpha) + \mathcal{O}((\delta \alpha)^3),\\
    \mathcal{I}(\alpha)&:=\int \frac{(\partial_\alpha p(\alpha))^2}{p(\alpha)}d\mu.
\end{align}

\subsection{Examples of structural $f$-divergences}

We show some examples of the structural $f$-divergences.
Let $D_{\mathrm{KL}}(P\|Q):=\int p\ln(p/q)d\mu$ be the Kullback-Leibler divergence.
\begin{itemize} 
    \item Squared Hellinger distance:
        \begin{align}
            H^2(P,Q):=\frac{1}{2}\int (\sqrt{
            p}-\sqrt{q})^2d\mu=1-\int \sqrt{pq}d\mu.
        \end{align}
         From Eqs.~\eqref{eq:def_binary_div} and~\eqref{eq:def_ex_div}, we obtain 
        \begin{align}
           D_{f}^{\mathrm{str}}(P,Q)= \sqrt{1-BC(P,Q)^2},
        \end{align}
        where $BC(P,Q):=\int \sqrt{pq}d\mu$ is the Bhattacharyya coefficient.

\item Jensen-Shannon divergence: 
         \begin{align}
             \mathrm{JS}(P,Q):=\frac{1}{2}(D_{\mathrm{KL}}(P\|M) + D_{\mathrm{KL}}(Q\|M)),
         \end{align}
         where $M:=(P+Q)/2$.
        From Eq.~\eqref{eq:def_binary_div}, it follows that
        \begin{align}
            d_f(s)=\frac{1+s}{2}\ln(1+s)+\frac{1-s}{2}\ln(1-s).
        \end{align}
    \item Jeffrey's divergence: 
        \begin{align}
            \mathrm{J}(P,Q):=\frac{1}{2}(D_{\mathrm{KL}}(P\|Q) + D_{\mathrm{KL}}(Q\|P)).
        \end{align}
        From Eq.~\eqref{eq:def_binary_div}, it follows that
        \begin{align}
            d_f(s)=s\ln \left(\frac{1+s}{1-s}\right)=2s\mathrm{artanh}(s).
        \end{align}
\end{itemize}

\section{Conclusion}
We have established a theoretical framework centered on the structural $f$-divergence to quantify the information-theoretic discrepancies within parameterized quantum circuits on both the parameter space and the induced measure of the unitary group.
By establishing bounds on the parameter space, we provide a direct means to evaluate how specific sampling strategies, such as initialization schemes or parameter-wise noise, shape the resulting cost function landscape.
Our primary contribution lies in the derivation of analytically tight trade-off relations. We have demonstrated that the structural $f$-divergence rigorously bounds both the expectation value of the gradient magnitudes and the moments of the cost function. Furthermore, by identifying specific ansatz architectures that satisfy the equality conditions, our results provide the most stringent limits possible on how statistical distributions dictate landscape behavior.
These findings suggest that structural $f$-divergence is not merely a theoretical construct but a fundamental metric for quantum model selection and architecture design. 
We anticipate that this metric will serve as a guiding principle for the optimal design of quantum circuits.

\begin{acknowledgements}
This work was supported by the Japan Society for the Promotion of Science KAKENHI Grant Numbers JP24K03008 and JP26K02998.
\end{acknowledgements}

\appendix
\begin{widetext}

\section{Proof of Eq.~\eqref{eq:properties_div} \label{sec:proof_properties}}

\begin{lemma}
    $d_f$ is monotonically increasing in $[0,1]$.
    \label{lem1}
\end{lemma}
\begin{proof}   
    Letting $v(s):=(1+s)(1-s)=2/(1-s)-1$ for $s\in (0,1)$, and differentiating Eq.~\eqref{eq:def_binary_div} yields
    \begin{align}
        d_f^\prime(s)=\frac{1}{2}\left(f(v(-s))-f(v(s))\right)+ \frac{1}{1-s}f^\prime(v(s))-\frac{1}{1+s}f^\prime(v(-s)).
        \label{eq:diff_df}
    \end{align}
Since $f(v(-s))-f(v(s))> f^\prime(v(s))(v(-s)-v(s))=2 f^\prime(v(s))(1/(1+s)-1/(1-s))$ holds from the strict convexity of $f$, 
substituting this relation into Eq.~\eqref{eq:diff_df} yields
\begin{align}
    d_f^\prime(s)>\frac{1}{1+s}\left(f^\prime(v(s))-f^\prime(v(-s))\right)> 0,
    \label{eq:d_f_monotonicity}
\end{align}
where we use the monotonicity of $f^\prime(s)$ for $v(s)> v(-s)$ for $s\in (0,1)$.
\end{proof}
We prove Eq.~\eqref{eq:properties_div}. The first equality follows from the domain of $d_f$.
The second equality follows from the property of the $f$-divergence and $d_f(0)=0$ since $f(1)=0$. The third equality follows from the definition. The fourth equality follows from the data processing inequality for the $f$-divergence and Lemma~\ref{lem1}.

\section{Proof of Eq.~\eqref{eq:main_result1} \label{sec:proof_main}}

\subsection{Lemmas}

Before the proof, we prove the following lemmas. 
\begin{lemma}
For any function $f\in \mathcal{F}$, 
   \begin{align}
        d_f(t)=\inf_{d_{\mathrm{TV}}(P,Q)=t}\tilde{D}_f(P,Q).
    \end{align}
    The infimum is attained by $(P_B, Q_B)$ defined by Eq.~\eqref{eq:def_two_element}.
    \label{lem2}
\end{lemma}
    The proof for differentiable $f$ is provided in Ref.~\cite{sason2015tight, gilardoni2010pinsker, CRMATH_2006__343_11-12_763_0}. For completeness, we briefly present an alternative proof below. For the full derivation, see the original papers.
    \begin{proof}
    Since the case where the $f$-divergence is the total variation distance is trivial as $d_f(t)=t$, we focus on the cases where $f(1)=0$ and $f''(x)>0$ in Eq.~\eqref{eq:def_set_function}.
    By differentiating Eq.~\eqref{eq:diff_df}, the convexity of $f$ implies that for any $s\in (0,1)$, 
    \begin{align}
        d_f^{\prime\prime}(s)= \frac{2}{(1-s)^3}f^{\prime\prime}(v(s))+\frac{2}{(1+s)^3}f^{\prime\prime}(v(-s))>0.
        \label{eq:d_f_convex}
    \end{align}
    Let $g:=d_f^{-1}$. For $x=d_f(s)$, the second derivative of $g$ is given by
    \begin{align}
        \frac{d^2}{dx^2}g(x)=-\frac{d_f^{\prime\prime}(s)}{(d_f^\prime(s))^3}.
    \end{align}
    Combining this with Eqs.~\eqref{eq:d_f_monotonicity} and~\eqref{eq:d_f_convex}, it follows that $g$ is concave.
    By substituting $s=|p-q|/(p+q)\in [0,1]$ into the identity $g(d_f(s))=s$, multiplying by $(p+q)/2$ and integrating with respect to the dominating measure $\mu$, we obtain
    \begin{align}
        \int \frac{(p+q)}{2}g\left(d_f\left(\frac{|p-q|}{p+q}\right)\right)d\mu=\frac{1}{2}\int |p-q| d\mu=d_{\mathrm{TV}}(P,Q)=t.
    \end{align}
    Since $g$ is concave, the Jensen's inequality implies
      \begin{align}
        t\le g\left( \int\frac{(p+q)}{2}d_f\left(\frac{|p-q|}{p+q}\right)d\mu\right)=g(\tilde{D}_f(P,Q)),
    \end{align}
    where we use Eq.~\eqref{eq:def_binary_div} in the last equality. From the montonicity of $d_f$ (Lemma~\ref{lem1}), we obtain 
    \begin{align}
        d_f\left(t\right)\le\tilde{D}_f(P,Q).
    \end{align}
    From $t=d_{\mathrm{TV}}(P_B,Q_B)=|r|$ and Eq.~\eqref{eq:def_binary_div}, the infimum is attained by $(P_B, Q_B)$.
\end{proof}

\begin{lemma}
    Let $X\geq 0$ be a random variable, and let $X_{\max}\geq 0$ be the maximum value of $X$.
    For probability measures $P$ and $Q$ for the random variable $X$, 
    \begin{align}
        |\mathbb{E}_P[X]-\mathbb{E}_Q[X]|\le X_{\max} d_{\mathrm{TV}}(P,Q).
        \label{eq:TVdist_ub}
    \end{align}
    The equality holds if 
    \begin{align}
        P_B(X=0)=\frac{1-r}{2}, \;P_B(X=X_{\max})=\frac{1+r}{2},\\
        Q_B(X=0)=\frac{1+r}{2}, \;Q_B(X=X_{\max})=\frac{1-r}{2}.
    \end{align}
    \label{lem3}
\end{lemma}

\begin{proof}
    Since $\int_{p\geq q} p d\mu + \int_{q>p} pd\mu=1$, we obtain $\int_{p\geq q} (p-q)d\mu=\int_{q>p} (q-p)d\mu$. This relation yields
    $d_{\mathrm{TV}}(P,Q)=\int_{p\geq q} (p-q)d\mu=\int_{q>p} (q-p)d\mu$.
    By combining $\mathbb{E}_P[X]-\mathbb{E}_Q[X]=\int x(p-q)d\mu\le \int_{p\geq q} x(p-q)d\mu \le X_{\max} d_{\mathrm{TV}}(P,Q)$ and $\mathbb{E}_P[X]-\mathbb{E}_Q[X]\geq  \int_{q>p} x(p-q)d\mu \geq -X_{\max} d_{\mathrm{TV}}(P,Q)$, we obtain Eq.~\eqref{eq:TVdist_ub}. From $|\mathbb{E}_P[X]-\mathbb{E}_Q[X]|=|r|X_{\max}$ and $d_{\mathrm{TV}}(P_B,Q_B)=|r|$, the equality condition follows.
\end{proof}

\begin{lemma}
    Let the index $j$ denote the $m$-th gate in layer $l$. 
For the gradient with respect to 
    $\theta_j$, the following inequality holds:
    \begin{align}
        |\partial_j \braket{O}|\le 2\|H_j\|_R\|O\|_\infty.
        \label{eq:grad_ub}
    \end{align} 
    The equality is achieved by the ansatz $\mathcal{C}_{1,1}$ at $\theta=\pi/2$.
    
    \label{lem4}
\end{lemma}
\begin{proof}

We define $U^{-}$ as the product of unitary operators with indices smaller than $j$, and $U^{+}$ as the product of those with indices larger than $j$.
The gradient of $U$ is calculated as 
\begin{align}
    \partial_j U(\bm{\theta})=-iU^{+} H_j U^{-}.
\end{align}
Letting $O^{+}:={(U^{+})}^\dagger OU^{+}$ and $\rho^{-}:={U^{-}} \rho {(U^{-})}^\dagger$, we obtain
    \begin{align}
        \partial_j \braket{O}=\mathrm{Tr}[\rho \partial_j U(\bm{\theta})^\dagger O U(\bm{\theta})]+\mathrm{Tr}[\rho U(\bm{\theta})^\dagger O \partial_j U(\bm{\theta})]=i\mathrm{Tr}[[H_j,O^{+}]\rho^{-}].
    \end{align}
By applying $|\mathrm{Tr}[A]|\le \|A\|_1$ and the H\"{o}lder's inequality, it follows that
\begin{align}
    &|\mathrm{Tr}[[H_j,O^{+}]\rho^{-}]|=|\mathrm{Tr}[[H_j-a,O^{+}]\rho^{-}]|\le \|[H_j-a,O^{+}]\rho^{-}\|_1\nonumber\\
    &\le (\|(H_j-a)O^{+}\|_\infty+\|O^{+}(H_j-a)\|_\infty)\|\rho^{-}\|_1\le2\|(H_j-a)\|_\infty\|O\|_\infty.
\end{align}
Substituting $a=(\lambda^H_{\max}+\lambda^H_{\min})/2$ into this relation, we obtain Eq.~\eqref{eq:grad_ub}. 

We next prove the equality condition. Recall that $\mathcal{C}_{1,1}$ is defined by Eq.~\eqref{eq:equality_cond_circuit}. From $[H,O]=i\sigma_y$, $U(\theta)^\dagger \sigma_y U(\theta)=\sigma_y \cos\theta+\sigma_x\sin\theta$, $\braket{\mathrm{init|\sigma_x|\mathrm{init}}}=1$, and $\braket{\mathrm{init|\sigma_y|\mathrm{init}}}=0$, we obtain
\begin{align}
    \partial O=i\mathrm{Tr}[\rho U(\theta)^\dagger [H,O] U(\theta)]=-\sin\theta.
    \label{eq:equality_cond_theta}
\end{align}
By combining $H=\sigma_z/2$, $O=\sigma_x$ with $\|\sigma_z\|_R=1$, $\|\sigma_x\|_\infty=1$, it follows that $|\partial O|=1=2\|H\|_R \|O\|_\infty$ at $\theta=\pi / 2$.
\end{proof}

\subsection{Proof}

For a random variable $0\le X\le X_{\max}$, by applying Lemma~\ref{lem1}--Lemma~\ref{lem3}, it follows that
\begin{align}
    \tilde{D}_f(P,Q)\geq d_f(d_{\mathrm{TV}}(P,Q))\geq d_f\left(\frac{|\mathbb{E}_P[X]-\mathbb{E}_Q[X]|}{X_{\max}}\right).
    \label{eq:f_div_binary_div}
\end{align}
Substituting $X=|\partial_j \braket{O}|$, $P=P_{\Theta}$, and $Q=Q_{\Theta}$ into Eq.~\eqref{eq:f_div_binary_div} and using Lemma~\ref{lem1} and Lemma~\ref{lem4}, we obtain
\begin{align}
    \tilde{D}_f(P_{\Theta},Q_{\Theta})\geq d_f\left(\frac{ \left|\mathbb{E}_{P_{\Theta}}[\:|\partial_j \braket{O}|\:]-\mathbb{E}_{Q_{\Theta}}[\:|\partial_j \braket{O}|\:]\right|}{2\|H_j\|_R\|O\|_\infty}\right).
    \label{eq:f_div_expectation}
\end{align}
When $f(x)=|x-1|/2$, the same inequality follows from $d_f(s)=s$, Lemma~\ref{lem3} and~\ref{lem4}.
Since Eq.~\eqref{eq:f_div_expectation} holds for any $f\in\mathcal{F}$, from Lemma~\ref{lem1}, we obtain Eq.~\eqref{eq:main_result1}.

We next prove that Eq.~\eqref{eq:main_result1} is tight. We show an example of equality condition. From Eqs.~\eqref{eq:f_div_binary_div} and~\eqref{eq:f_div_expectation}, the equality condition for Eq.~\eqref{eq:main_result1} holds if and only if all equality conditions for Lemma~\ref{lem2} through~\ref{lem4} are satisfied.
The  ansatz  $\mathcal{C}_{1,1}$ satisfies the equality conditions for Lemma~\ref{lem4}.
From Lemma~\ref{lem2},~\ref{lem3} and Eq.~\eqref{eq:equality_cond_theta}, the equality conditions for the probability measures are given by Eq.~\eqref{eq:equality_cond_grad}, regardless the choice of $f$.

\section{Proof of Eq.~\eqref{eq:main_result2}\label{sec:proof_main2}}

Before the proof, we prove the following lemmas.
\subsection{Lemmas}
\begin{lemma}
    Let $-X_{\max} \le X\le X_{\max}$ be a random variable.
    For probability measures $P$ and $Q$ for the random variable $X$, 
    \begin{align}
        |\mathbb{E}_P[X]-\mathbb{E}_Q[X]|\le 2X_{\max} d_{\mathrm{TV}}(P,Q).
        \label{eq:TVdist_ub2}
    \end{align}
    The equality holds if 
    \begin{align}
        P_B(X=-X_{\max})=\frac{1-r}{2}, \;P_B(X=X_{\max})=\frac{1+r}{2},\\
        Q_B(X=-X_{\max})=\frac{1+r}{2}, \;Q_B(X=X_{\max})=\frac{1-r}{2}.
    \end{align}
    \label{lem5}
\end{lemma}
Applying Lemma~\ref{lem3} for a random variable $Y=X+X_{\max}$, the result immediately follows.

\begin{lemma}
    The following inequality holds.
    \begin{align}
        |\braket{O}|\le \|O\|_\infty.
        \label{eq:operator_norm_ub}
    \end{align}
    The equality is achieved by the ansatz $\mathcal{C}_{1,1}$ at $\theta=0$.
    \label{lem6}
\end{lemma}
\begin{proof}
    By applying $|\mathrm{Tr}[A]|\le \|A\|_1$ and the H\"{o}lder's inequality, it follows that 
    \begin{align}
        |\braket{O}|=|\mathrm{Tr}[U(\bm{\theta})\rho U(\bm{\theta})^\dagger O ]|\le \|U(\bm{\theta})\rho U(\bm{\theta})^\dagger\|_1\|O\|_\infty=\|O\|_\infty.
    \end{align}

    We next prove the equality condition for $O=\sigma_x$.
    From $U(\theta)^\dagger \sigma_x U(\theta)=\sigma_x\cos\theta-\sigma_y\sin\theta$, $\braket{\mathrm{init|\sigma_x|\mathrm{init}}}=1$, and $\braket{\mathrm{init|\sigma_y|\mathrm{init}}}=0$, we obtain 
    \begin{align}
        \braket{O}=\mathrm{Tr}[\rho U(\theta)^\dagger O U(\theta)]=\cos\theta.
        \label{eq:equality_cond_theta2}
    \end{align}
    Since $\|\sigma_x\|_\infty=1$, the equality holds at $\theta=0$.
\end{proof}
\subsection{Proof}
In the following, we first provide a proof for the probability measures $(P_{\mathcal{U}},Q_{\mathcal{U}})$. Letting $X=\braket{O}^k$, the proof is analogous to that of Eq.~\eqref{eq:main_result1}.
We consider the case where $k$ is even. 
Since $\braket{O}^k\geq 0$, by applying Lemma~\ref{lem1},~\ref{lem2},~\ref{lem3} and~\ref{lem6}, we obtain Eq.~\eqref{eq:main_result2} for $C(k)=1$.
Consider the  ansatz  $\mathcal{C}_{1,1}$. 
From Eq.~\eqref{eq:equality_cond_theta2}, Lemma~\ref{lem2},~\ref{lem3}, and~\ref{lem6}, the equality holds for Eq.~\eqref{eq:equality_cond_k_even}.
We next consider the case where $k$ is odd. 
Following a similar procedure in Eq.~\eqref{eq:f_div_binary_div}, by applying Lemma~\ref{lem1},~\ref{lem2},~\ref{lem5} and~\ref{lem6}, we obtain Eq.~\eqref{eq:main_result2} for $C(k)=2$.
From Eq.~\eqref{eq:equality_cond_theta2}, Lemma~\ref{lem2},~\ref{lem5} and~\ref{lem6}, the equality holds for Eq.~\eqref{eq:equality_cond_k_odd}.

The proof of Eq.~\eqref{eq:main_result2} for $(P_{\Theta},Q_{\Theta})$ follows from Eqs.~\eqref{eq:data_processing_str} and~\eqref{eq:induced_relation}. A similar argument applies to the proof of the equality conditions.

\section{Equality condition for $n$-qubit circuit  \label{sec:equaility_nqubit}}

Let $\ket{v}:=\frac{1}{\sqrt{2}}(\ket{0}+\ket{1})$, and let $I_{2}$ be the $2\times 2$ identity operator.
Consider  an $n$-qubit and one-layer ansatz that corresponds to the case of $L=1$ and $M=n$ in Eq.~\eqref{def:unitary_gates} , where parameters are given by $\bm{\theta}\in [0,2\pi]^n$:
\begin{align}
    &H_j = I_2^{\otimes (n-j)}\otimes\frac{\sigma_z}{2}\otimes I_2^{\otimes (j-1)} \; \text{ for $1\le j \le n$},\nonumber\\ 
    &O=\sigma_x^{\otimes n}, \;\ket{\mathrm{init}}:=\ket{v}^{\otimes n}.
    \label{eq:equality_nqubit}
\end{align}
For simplicity, we consider the case where $j=1$.
As in Eqs.~\eqref{eq:equality_cond_theta} and~\eqref{eq:equality_cond_theta2}, we obtain 
\begin{align}
    \partial_1 O&=-\sin\theta_1\prod_{j=2}^n \cos\theta_j, 
    \label{eq:equality_theta_nqubit} \\\
    \braket{O}&=\prod_{j=1}^n \cos\theta_j.
\end{align}
Noting that $2\|H_1\|_R=\|O\|_\infty=1$, the equality holds in Lemma~\ref{lem4} for $\bm{\theta}=\left(\frac{\pi}{2}, 0, 0, \ldots, 0\right)$ and in Lemma~\ref{lem6} for $\bm{\theta}=\bm{0}$, respectively. 
For $(P_{\Theta},Q_{\Theta})$, the equality condition for Eq.~\eqref{eq:main_result1} and Eq.~\eqref{eq:main_result2} for even $k$ is given by
\begin{align}
    P_{\Theta}^B(\bm{\theta}=\bm{0})=\frac{1-r}{2}, \;P_{\Theta}^B\left(\bm{\theta}=\left(\frac{\pi}{2}, 0, 0, \ldots, 0\right)\right)=\frac{1+r}{2}, \nonumber \\
    Q_{\Theta}^B(\bm{\theta}=\bm{0})=\frac{1+r}{2}, \;Q_{\Theta}^B\left(\bm{\theta}=\left(\frac{\pi}{2}, 0, 0, \ldots, 0\right)\right)=\frac{1-r}{2}.
    \label{eq:equality_grad_nqubit}
\end{align}
The equality condition of Eq.~\eqref{eq:main_result2} for odd $k$ is given by
\begin{align}
    P_{\Theta}^B(\bm{\theta}=\bm{0})=\frac{1-r}{2}, \;P_{\Theta}^B\left(\bm{\theta}=\left(\pi, 0, 0, \ldots, 0\right)\right)=\frac{1+r}{2}, \nonumber \\
    Q_{\Theta}^B(\bm{\theta}=\bm{0})=\frac{1+r}{2}, \;Q_{\Theta}^B\left(\bm{\theta}=\left(\pi, 0, 0, \ldots, 0\right)\right)=\frac{1-r}{2}.
    \label{eq:equality_grad_nqubit2}
\end{align}
The equality conditions of Eq.~\eqref{eq:main_result2} with $(P_{\mathcal{U}},Q_{\mathcal{U}})$ can be derived in a similar manner:
For even $k$, we obtain
\begin{align}
    P_{\mathcal{U}}^B(U=I_2^{\otimes n})=\frac{1-r}{2}, \;P_{\mathcal{U}}^B\left(U=I_2^{\otimes(n-1)}e^{-i\pi H_1/2}\right)=\frac{1+r}{2}, \nonumber \\
    Q_{\mathcal{U}}^B(U=I_2^{\otimes n})=\frac{1+r}{2}, \;Q_{\mathcal{U}}^B\left(U=I_2^{\otimes(n-1)}e^{-i\pi H_1/2}\right)=\frac{1-r}{2}. 
    \label{eq:equality_grad_nqubit_unitary}
\end{align}
The equality condition for odd $k$ is given by
\begin{align}
     P_{\mathcal{U}}^B(U=I_2^{\otimes n})=\frac{1-r}{2}, \;P_{\mathcal{U}}^B\left(U=I_2^{\otimes(n-1)}e^{-i\pi H_1}\right)=\frac{1+r}{2}, \nonumber \\
    Q_{\mathcal{U}}^B(U=I_2^{\otimes n})=\frac{1+r}{2}, \;Q_{\mathcal{U}}^B\left(U=I_2^{\otimes(n-1)}e^{-i\pi H_1}\right)=\frac{1-r}{2}.
    \label{eq:equality_grad_nqubit2_unitary}
\end{align}

\section{Asymptotic behavior of $ D_{f}^{\mathrm{str}}(P,Q)$ \label{sec:asymptotic}}

Applying the Taylor expansion of $f(x)$ around $x=1$ and using $f(1)=0$, $\int pd\mu=1$, we obtain
\begin{align}
    \tilde{D}_f(P,Q)=f''(1)\int \frac{(p-q)^2}{p+q}d\mu +\mathcal{O}\left(\int (p-q)^4d\mu \right)=2f''(1)\Delta(P,Q)+\mathcal{O}\left(\int (p-q)^4d\mu \right) \ll 1.
\end{align}
From  Eq.~\eqref{eq:def_binary_div}, it follows that $d_f(s)= 2f''(1)s^2 + \mathcal{O}(s^4)$ for $|s|\ll 1$. Since $f''(1)> 0$ for $f\in \mathcal{F}$, we obtain 
\begin{align}
    D_{f}^{\mathrm{str}}(P,Q)=\sqrt{\Delta(P,Q)}+\mathcal{O}\left(\int (p-q)^3d\mu \right)
\end{align}
for arbitrary $f(x)$.

\end{widetext}

\end{document}